\documentclass[letterpaper]{article} 
\usepackage{aaai2026}  
\usepackage{times}  
\usepackage{helvet}  
\usepackage{courier}  
\usepackage[hyphens]{url}  
\usepackage{graphicx} 
\usepackage{natbib}  
\usepackage{caption} 
\usepackage{algorithm}
\usepackage{algorithmic}

\usepackage{newfloat}
\usepackage{listings}
\DeclareCaptionStyle{ruled}{labelfont=normalfont,labelsep=colon,strut=off} 
\floatstyle{ruled}
\newfloat{listing}{tb}{lst}{}
\floatname{listing}{Listing}
\title{
Public Goods Games in Directed Networks with Constraints on Sharing 
}

\author {
    Argyrios Deligkas\textsuperscript{\rm 1},
    Gregory Gutin\textsuperscript{\rm 1, 2},
    Mark Jones\textsuperscript{\rm 3},
    Philip R.\ Neary\textsuperscript{\rm 1},
    Anders Yeo\textsuperscript{\rm 4,5}
}
\affiliations {
    \textsuperscript{\rm 1}Royal Holloway University of London, UK\\
    \textsuperscript{\rm 2}Nankai University, China\\
    \textsuperscript{\rm 3}Independent Researcher, UK\\
    \textsuperscript{\rm 4}University of Southern Denmark, Denmark\\
    \textsuperscript{\rm 5}University of Johannesburg, South Africa\\
    \{Argyrios.Deligkas, G.Gutin, Philip.Neary\}@rhul.ac.uk, markelliotlloyd@gmail.com, yeo@imada.sdu.dk
}

\usepackage{amsmath}
\usepackage{amsfonts}
\usepackage{amsthm}
\usepackage{cleveref}
\usepackage{xspace}
\usepackage{xcolor}
\usepackage{tikz}
\usetikzlibrary{backgrounds}
\usepackage{todonotes}
\usepackage{nicefrac}
\usepackage{bm}
\usepackage{enumitem}

\usepackage{tabularx}
\usepackage{booktabs}
\usepackage{multirow}

\newcommand{\set}[1]{\left\{#1\right\}}

\newcommand{\2}{\vspace{0.2cm}}

\newcommand{\leaves}{\mathrm{Leaves}}
\newcommand{\roots}{\mathrm{Roots}}

\newcommand{\util}[1]{U_{#1}}

\newcommand{\IDr}{$1$-buyer}
\newcommand{\Dr}{buyer}
\newcommand{\kPNa}{$k$-pure-Nash}
\newcommand{\IPNa}{$1$-pure-Nash}
\newcommand{\PNa}{pure-Nash}

\NewDocumentCommand{\cc}{ O{} O{} m }{\mbox{%
    \expandafter\ifx\expandafter\relax\detokenize{#2}\relax\else{#2-}\fi%
    \textsf{#3}%
    \expandafter\ifx\expandafter\relax\detokenize{#1}\relax\else{-#1}\fi%
    }\xspace}

\newtheorem{theorem}{Theorem}

\newtheorem{definition}{Definition}

\newtheorem{conjecture}{Conjecture}

\newif\iflong
\newif\ifshort

\longtrue

\iflong
\else
 \shorttrue
\fi

\begin{document}

\maketitle

\begin{abstract}
In a public goods game, every player chooses whether or not to buy a good that all neighboring players will have access to.
We consider a setting in which the good is indivisible, neighboring players are out-neighbors in a directed graph, and there is a capacity constraint on their number, $k$, that can benefit from the good. 
This means that each player makes a two-pronged decision: decide whether or not to buy and, conditional on buying, choose which $k$ out-neighbors to share access.
We examine both pure and mixed Nash equilibria in the model from the perspective of existence, computation, and efficiency. 
We perform a comprehensive study for these three dimensions with respect to both sharing capacity ($k$) and the network structure (the underlying directed graph), and establish sharp complexity dichotomies for each.
\end{abstract}



\section{Introduction}

A public good is one that is both {\em non-rivalrous} and usually {\em non-excludable}: the good can be used by many players, even by those that did not contribute to the cost. 
This occurs when either (a) it is difficult, or impossible, to prevent the non-contributors to use the good -- think of public lighting -- or (b) the contributors are allowed to {\em choose to share} the good with a {\em subset} of their peers -- think e.g. of Netflix allowing multiple users associated with a single account. 
Hence, public goods also generate a coordination problem because all parties would prefer to free ride, i.e., every player wants to use the good, but they prefer someone else to pay its cost.

In this paper, we consider a variant of the public goods problem in social networks as considered first by \citet{BramoulleKranton:2007:JET}: there is an underlying graph, where nodes correspond to players, each player chooses whether to contribute or not, and every player can benefit from their neighbors that contributed.
We adopt the {\em limited} and {\em asymmetric shareability} point of view in public good games and we augment the original model in two ways.

First, under limited shareability players might face capacity constraints as in \citep{GerkeGutin:2024:JET,gutin2020uniqueness,GutinNeary:2023:DAM}.
Constraints on sharing add an additional layer of 
complexity because in order to benefit from someone else's contribution, it is not enough to be neighbours: one must also be nominated as a co-beneficiary. 
Consider for example an online service where account-sharing is possible, but where the platform allows at most $k+1$ users to simultaneously access an account: the owner of the account cannot legally share the account with more than $k$ of their friends.
Thus, the owner has to {\em choose} which $k$ friends to invite to the account.

Asymmetric shareability can be captured by modeling the societal structure by a directed graph~\citep{BAYER2023105720,LOPEZPINTADO2013160,PapadimitriouPeng:2023:GEB}.
A directed edge from node $i$ to $j$, allows contributing player $i$ to choose out-neighbour player $j$. 
This can be due to asymmetric preferences, or due to feasibility constraints.
Consider, for example, a group of people each of whom needs to reach to their work and has to decide whether to drive or not. 
Each player $j$ has to choose either to drive, or to get a lift from a player $i$ that chooses to drive {\em and} player $j$ is on their way to work.
Hence, anytime a player $i$, who has $k$ empty seats, drives, they offer a ride to $\min\set{k, d^+(i)}$ out-neighbours, where $d^+(i)$ is $i$'s out-degree.

While asymmetric shareability is more expressive and closer to reality, it creates additional complications for the game.
When the underlying network is undirected, pure strategy Nash equilibria are guaranteed to exist and are easy to find. 
When $k$ is greater than the max degree in the graph, equilibria correspond to independent dominating sets \cite{BramoulleKranton:2007:JET}.
When $k$ is less than the max degree a pure strategy Nash equilibrium always exists and one can be computed in polynomial time~\citep{GerkeGutin:2024:JET,gutin2020uniqueness,GutinNeary:2023:DAM}.
On the other hand, if the underlying social network is directed, matters are not so straightforward.
When sharing is {\em not} constrained \citep{LOPEZPINTADO2013160,BAYER2023105720,PapadimitriouPeng:2023:GEB}, pure strategy equilibria are described by what is known as a {\em kernel}, the analog of independent dominating sets for digraphs, which does not always exist.\footnote{\label{fn:kernel}Kernels are well-studied objects in graph theory but were introduced first in \cite{NeumannMorgenstern:1944:} as the solutions to cooperative games. 
\iflong
For example, when a two-sided matching market is formulated as a digraph \citep{BalinskiRatier:1997:,GutinNeary:2023:GEB,Maffray:1992:JCTB}, the collection of pairs that make up a stable matching comprise a kernel.\fi
}
Worse still, deciding if a digraph contains a kernel is NP-complete \citep{chvatal1973computational, PapadimitriouPeng:2023:GEB}.
However, the existence and complexity of equilibria in public good games with limited shareability on digraphs remained open. 

\subsection{Our Contribution}

We perform a comprehensive study on the complexity and the economic efficiency of Nash equilibria in public good games on directed graphs with sharing constraints. 
Our results revolve around the structure of the underlying digraph and the sharing parameter $k$. 

Our first set of technical results focuses on pure Nash equilibria. We observe that when $k=1$, equilibrium computation boils down to a graph-theoretic problem which closely related to what \citet{DG12} refer to as a {\em spanning galaxy}, a vertex-disjoint union of stars containing all vertices.
We use this connection to prove that deciding whether a pure Nash equilibrium exists is NP-complete in general (Thm.~\ref{thm:storng_vs_general}). However, when the underlying graph is strongly connected, then the problem can be solved in polynomial time (Thm.~\ref{thm:storng_vs_general}).
Unfortunately though, for every $k \geq 2$, the problem is NP-complete even on strongly connected digraphs (Thm.~\ref{thm:k>=2}).
On the positive side, in Thm.~\ref{thm:structural_results}, we identify some structural properties of the underlying digraph that guarantee the existence of a pure Nash equilibrium.

Then, we study the problem of computing a mixed Nash equilibrium, which is guaranteed to exist and we show a dichotomy with respect to $k$.
In Thm.~\ref{thm:k_one_poly} we derive a linear-time algorithm for the problem, when $k=1$. The dichotomy follows from previous results that showed PPAD-hardness for $k\geq 2$ \cite{PapadimitriouPeng:2023:GEB} even for {\em approximate} Nash equilibria \cite{DODINH2024106486}.

Our last set of results almost resolves completely the (pure) {\em Price of Stability} (PoS)~\cite{schulz2003performance} and {\em Anarchy} (PoA)~\cite{koutsoupias1999worst} of this type of game. 
When $k=1$ and a pure strategy equilibrium exists, Thm.~\ref{priceOfStabilityI} shows that PoS=1; in fact, the theorem exactly characterizes the number of buyers in every pure Nash equilibrium that maximizes the social cost.
For $k \ge 1$, in Thm.~\ref{thm:PoS} we prove an almost tight bound on PoS, again assuming that a pure Nash equilibrium exists: it is lower bounded by $k$ and upper bounded by $k +\frac{1}{k+1}$.
Next we move to Price of Anarchy and we prove that the pure PoA is exactly $k+1$ (Thm.~\ref{thm:anarchy}).
Interestingly, we show that the same bound, i.e., $k+1$ holds for PoS when we consider {\em mixed} Nash equilibria (Thm.~\ref{ThmSupportConj}).

\subsection{Further Related Work}

Our model is a special case of a {\em graphical game} \citep{kearns2013graphical} that has received significant attention over the years~\cite{deligkas2023tight}.
\citet{BramoulleKranton:2007:JET} initiated the study of public good provision in networks. 
They considered an undirected network and continuous action space.
This influential paper has generated a long line of follow up work \cite{Allouch:2015:JET,Allouch:2017:GEB,Baetz:2015:TE,BramoulleKranton:2014:AER,ElliottGolub:2019:JPE,KinatederMerlino:2017:AEJM,klimm2023complexity}.

\iflong
We have assumed that the public good is indivisible and that payoffs are of the ``best shot'' variety \citep{levit2018incentive,dall2011optimal,GaleottiGoyal:2010:RES,boncinelli2012stochastic}, in that all that is required to reach a bliss point is that one neighbor provides and offers to share.
But other payoff specifications are possible --- arguably any utility function that is non-decreasing in the number of neighbors who produce seems plausible --- see \cite{PapadimitriouPeng:2023:GEB} for other specifications.
\fi

Given their importance in society, the ability to compute equilibria in environments with public goods games has become a major topic of interest.
\citet{yu2020computing} show that finding a pure Nash equilibrium in a discrete version of the public goods game with heterogeneous agents players is NP-hard.
\citet{yang2020refined} consider a binary action public goods game and show that computation of equilibria are, in general, NP-hard but become polynomial time solvable when the underlying network is restricted to some special domains, e.g., those with bounded treewidth.
\citet{gilboa2022complexity} answer a question of \citet{PapadimitriouPeng:2023:GEB} and prove that for some simple ``best-response pattern'' the problem of determining whether a non-trivial pure Nash equilibrium exists is NP-complete; they further find a polynomial time algorithm for some specific simple pattern.
\citet{gilboa2023characterization} and \citet{klimm2023complexity} present alternative proofs that extend this result to show that the problem is NP-complete for patterns that are finite and non-monotone.

\ifshort
\smallskip
\noindent
{\em 
Some of the complete proofs of our results, identified with $(\star)$, are deferred to the full version of our paper.
}
\fi

\section{Model and preliminaries}\label{sec:model}
An instance of our problem is described by a directed graph $D= (V, A)$, where $V$ is the set of $n$ vertices and $A$ is the arc set. For every vertex $i \in V$, we denote its outgoing degree by $d^+(i)$.
The vertices correspond to players and arcs represent directed connections between pairs of them. 
Each player has two options: {\em buy} or {\em abstain}. 
If a player buys, they are required to choose $k$ of their out-neighbors, or all of their out-neighbors if $k \ge d^+(i)$.
To simplify notation  we write $k_i$ to mean $\min\set{k, d^+(i)}$.

%
Formally, the set of pure strategies of player $i$ is given by $X_i = \set{\neg} \cup \binom{N^+(i)}{k_i}$, where $\neg$ means that player $i$ does not buy and each element in the collection of sets $\binom{N^{+}(i)}{k_i}$ conveys that player $i$ buys and chooses a particular subset of size $k_i$ of their out-neighbours. In particular, $\binom{\emptyset}{0}=\{\emptyset\}.$

A pure strategy profile is a vector $\boldsymbol{x}  = (x_1, \dots x_n) \in \times_{j=1}^{n}X_j$ that specifies a pure strategy for each player. 
With price of buying given by $c \in (0,1)$, utilities are formally defined as follows.
\begin{equation*}
\util{i} (\boldsymbol{x}) = \left\lbrace
  \begin{array}{r l}
    -1, & \text{ if } x_i = \neg \text{ and } i \not\in x_j \text{ for all } j \in N^{-}(i),\\
    0, & \text{ if } x_i = \neg \text{ and } i \in x_j \text{ for a } j \in N^{-}(i),\\
    -c, & \text{ if } x_i \neq \neg.
  \end{array}
\right.
\end{equation*}

Observe that maximizing utility is equivalent to minimizing cost.
Moreover, note that the utility of a player that buys does not depend on who they nominate. 

\subsection{Questions of interest}
We focus on existence, computation, and the (relative) efficiency and inefficiency of equilibria.
\iflong
The latter is done by benchmarking the best and worst equilibrium outcome against the optimal outcome of the system as a whole. \fi

From the perspective of player $i$, a pure strategy profile $\boldsymbol{x}$ can be written as $(x_i, \boldsymbol{x}_{-i})$, where $x_i\in X_i$ is player $i$'s strategy and $\boldsymbol{x}_{-i}$ is the $(n-1)$-dimensional vector listing the behavior of all players other than $i$.
With this, pure Nash equilibria are defined as follows.
\begin{definition}\label{def:pureNE}
Strategy profile $\boldsymbol{x}^* = (x_1^*, \dots x_n^*)$ is a {\em pure Nash equilibrium} if for every $i = 1, \dots, n$, and every $x_{i} \in X_{i}$, it holds that 
$\util{i}(x_i^*, \boldsymbol{x}^*_{-i}) \geq \util{i}(x_i, \boldsymbol{x}^*_{-i})$.
\end{definition}

\iflong
Pure Nash equilibria do not always exist even when sharing is unconstrained.
One immediate barrier --- and one that will also play an important role in our more general model --- is odd cycles.
To see this, consider the three-player example in Figure~\ref{fig:noPure} taken from \cite{PapadimitriouPeng:2023:GEB}.
The three players are labeled $x, y$, and $z$ and the underlying network on which they are arranged is a directed cycle.
Sharing is unconstrained which is equivalent to $k=1$ in our model (since $k$ is equal to each player's out-degree in the network).
{To see that there is no pure strategy equilibrium, suppose to the contrary that there is one.
Note that in a pure strategy equilibrium at least one player must buy, and suppose that $x$ is some such player.
It follows that the best-response for $y$, who is $x$'s out-neighbor, is not to buy.
But since $y$ will not buy, $z$ must buy and as such it is optimal for $x$ not to buy.}

\begin{figure}[htb]
\begin{center}
\tikzstyle{vertexWC}=[circle, draw, minimum size=10pt, scale=0.75, inner sep=0.8pt]
\begin{tikzpicture}[scale=0.3]
  \node (x) at (1,3) [vertexWC]{$x$};
  \node (y) at (7,3) [vertexWC]{$y$};
  \node (z) at (4,7) [vertexWC]{$z$};
  \draw[->, line width=0.03cm] (x) -- (y);
  \draw[->, line width=0.03cm] (y) -- (z);
  \draw[->, line width=0.03cm] (z) -- (x);
\end{tikzpicture} 
\caption{No pure strategy equilibrium when $k=1$.}
\label{fig:noPure}
\end{center}
\end{figure}
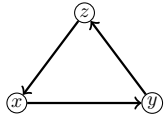
\fi

\ifshort
Pure Nash equilibria do not always exist even when sharing is unconstrained, i.e., $k\ge \max_i d^+(i)$.
\fi
If we allow for mixed strategies then existence is assured.
Formally, a mixed strategy for player $i$, denoted $\sigma_i$, is a probability distribution over their set of pure strategies.
That is, $\sigma_i(x_i) \geq 0$ for all $x_i \in X_i$ and $\sum_{x_i\in X_i} \sigma_i(x_i) =1$.
As with pure strategy profiles, a given mixed strategy profile $\boldsymbol{\sigma}  = (\sigma_1, \dots \sigma_n)$ can be viewed from the perspective of some player $i$ as $(\sigma_i, \boldsymbol{\sigma}_{-i})$.
The following extends the notion of equilibrium to mixed strategies.
(As is standard, we abuse notation somewhat and also use $\util{i}$ to denote player $i$'s utility defined on the space of mixed strategies.)
\begin{definition}\label{def:mixedNE}
A strategy profile $\boldsymbol{\sigma}^* = (\sigma_1^*, \dots \sigma_n^*)$ is a mixed strategy Nash equilibrium if for every $i = 1, \dots, n$, and every $\sigma_{i}$, it holds that
$\util{i}(\sigma_i^*, \boldsymbol{\sigma}^*_{-i}) \geq \util{i}(\sigma_i, \boldsymbol{\sigma}^*_{-i})$.
\end{definition}


Given a digraph $D$,
we denote $pn_k(D)$ and $Pn_k(D)$ the minimum and maximum, respectively, number of buyers in a pure strategy Nash equilibrium (conditional on existence) and we write $b_k(D)$ to denote the smallest number of buyers needed such that everyone either buys or is chosen by an in-neighboring buyer.
The (pure) price of stability (PoS) and the (pure) price of anarchy (PoA) are defined by comparing most/least efficient equilibrium against the most efficient outcome.
\footnote{Observe, the most efficient outcome is equivalent to minimum the number of buyers required to guarantee that every player has access to the good.}
Formally, both are given as follows, where suprema are taken for fixed $k$ over all $D$ that have a pure Nash equilibrium. 
\begin{equation*}
PoS_k = \sup_{D}\frac{pn_k(D)}{b_k(D)}; \hspace{.4in} PoA_k = \sup_{D} \frac{Pn_k(D)}{b_k(D)}
\end{equation*}

The mixed price of anarchy and mixed price of stability are then defined analogously
except that mixing is allowed.
If mixing is permitted, then by definition the efficiency attainable, both in and out of equilibrium, can only increase since optimization occurs over a strictly larger set.


\section{The complexity of pure Nash equilibria}

We begin this section by showing how all pure strategy outcomes in our model can be described in purely graph-theoretic terms;
we follow standard terminology on digraphs, see \citet{bang2008digraphs}.

Let $D$ be a digraph and let $k\ge 1$ be an integer.
We say that a subset of vertices S is a $k$-buyer set if by picking $k_i$ arcs out of every vertex $i$ in $S$ we can obtain a subdigraph with no isolated vertex in $V \setminus S$.
If a subdigraph, $R$, has no isolated vertex in $V \setminus S$ and is obtained by adding 
$k_i$ arcs out of every vertex in $S$, then we call $R$ an {\em extension} of $S$.
We will refer to the vertices in $S$ as buyers, the vertices in $V \setminus S$ as non-buyers, and the $k_i$ arcs out of each vertex $v \in S$ as an indication of which non-buyers will be chosen by $v$.
We say that a $k$-buyer set $S$ in $V$ is {\em \kPNa{}} if $S$ is independent in the extension $R$ (note it need not be independent in $D$) and by picking $k_i$ 
arcs out of every vertex $i$ in $S$ we obtain a subdigraph without any isolated vertex in $V \setminus S$.
When $k\ge \max_i d^{+}(i)$,  
a pure strategy Nash equilibrium is a {\em kernel} in the digraph $D$ (recall Footnote~\ref{fn:kernel}).

Whenever the value $k$ is clear from context, we will simply refer to each the above as {\em \Dr{}} sets and  {\em \PNa{}} sets respectively.
The following example illustrates buyer sets and pure Nash sets.

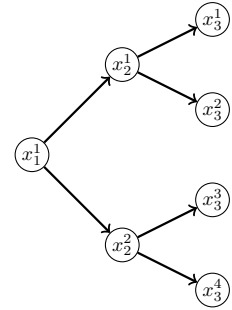
\begin{figure}[htb]
\begin{center}
\tikzstyle{vertexWC}=[circle, draw, minimum size=10pt, scale=0.75, inner sep=0.8pt]
\begin{tikzpicture}[scale=0.3]
  \node (x11) at (1,8) [vertexWC]{$x_1^1$}; 

  \node (x21) at (5,12) [vertexWC]{$x_2^1$};
  \node (x22) at (5,4) [vertexWC]{$x_2^2$};
  \draw[->, line width=0.03cm] (x11) -- (x21); 
  \draw[->, line width=0.03cm] (x11) -- (x22); 

  \node (x31) at (9,14) [vertexWC]{$x_3^1$};
  \node (x32) at (9,10) [vertexWC]{$x_3^2$};
  \node (x33) at (9,6) [vertexWC]{$x_3^3$};
  \node (x34) at (9,2) [vertexWC]{$x_3^4$};
  \draw[->, line width=0.03cm] (x21) -- (x31);
  \draw[->, line width=0.03cm] (x21) -- (x32); 
  \draw[->, line width=0.03cm] (x22) -- (x33);
  \draw[->, line width=0.03cm] (x22) -- (x34); 
\end{tikzpicture}
\caption{If $k=2$ then $\{x_1^1,x_2^1,x_2^2\}$ is a \Dr{} set  but not a \PNa{} set. In fact, the unique \PNa{} set is $\{x_1^1,x_3^1,x_3^2,x_3^3,x_3^4\}$.}
\label{Pic1}
\end{center}
\end{figure}

We note that, for any digraph $D$ and any sharing capacity $k$, there always exists a \Dr{} set because the entire vertex set of the digraph is one such a set.
\ifshort
However, even with unconstrained sharing, pure strategy equilibria are not guaranteed to exist; think of a directed cycle of length 3.
\fi
\iflong
However, as we saw from the example in Figure~\ref{fig:noPure}, pure strategy equilibria are not guaranteed to exist (even when sharing is not constrained).
\fi

There is a clear dichotomy in the complexity of determining the existence of a \PNa{} equilibrium for two cases. 
The first is $k = 1$ and the second is when $k \ge 2$.
We deal with each in turn, beginning with the simplest case where $k=1$.
We conclude this section by presenting some sufficient conditions for the existence of pure Nash equilibria.

\subsection{Pure equilibria when $k =1$}

For $k=1$, we observe that both equilibrium existence {\em and} the computation of deciding on existence depend on the structure of the underlying digraph. 
To obtain our results, we require some notions and tools from graph theory.

A digraph $D$ is {\em strong} if for every pair of vertices $i,j$ of $D$, there exists a directed path from $i$ to $j$ and a directed path from $j$ to $i$.
A {\em directed star}, or simply a star, consists of a central vertex, called the {\em root}, and at least one {\em leaf} vertex, with arcs from the root towards all leaves.
In \citet{DG12}, a {\em galaxy} is defined as vertex-disjoint union of stars and a {\em spanning galaxy} is a galaxy containing all vertices in the digraph.
We consider the reverse of a galaxy, that we call an {\em r-galaxy}, which is a galaxy with all arcs reversed.
If $R$ is an r-galaxy in a digraph $D$, then let $\roots(R)$ denote all the roots of the stars and let $\leaves(R)$ denote all the leaves (i.e., non-roots) in the stars. 

Before stating our first result, we introduce the following notation.
Given a digraph $D$, let $\delta^+(D)$ denote the minimum out-degree of any vertex in $D$.

\begin{theorem} \label{thmEquivI}
Let $D$ be a digraph with $\delta^+(D) \geq 1$.
If $R$ is a spanning r-galaxy in $D$ then $\leaves(R)$ is a \IPNa{} set in $D$.
Furthermore, if $S$ is a \IPNa{} set in $D$ then there exists a spanning r-galaxy, $R$, in $D$, such that $\leaves(R)=S$.
\end{theorem}

\begin{proof}
First let $R$ be a spanning r-galaxy in $D$ and let $S=\leaves(R)$.
We note that $S$ is a \PNa{} set in $D$, as $R$ is an extension of $S$.

Now let $S$ be a \PNa{} set in $D$ and let $R$ be the subdigraph obtained by adding one arc out of each vertex of $S$, such that 
$S$ remains an independent set and there are no isolated vertices. We note that $R$ is now a spanning r-galaxy with $\leaves(R)=S$, as desired.
\end{proof}

The qualifier that $\delta^+(D) \geq 1$ is important for Theorem~\ref{thmEquivI}, because, if some vertex in $D$ has out-degree zero, then it may still belong to a \IPNa{} set, but it can never be a leaf in a spanning r-galaxy.

We will utilize the following useful results from \citet{DG12}, which were proved for spanning galaxies and therefore trivially hold for spanning r-galaxies.

\begin{theorem}[\citet{DG12}]\label{DG12results}
The following holds.
\begin{description}
\item[(a)] Deciding whether a digraph $D$ has a spanning r-galaxy is NP-complete, even when restricted to digraphs which are acyclic, planar, bipartite, subcubic, or with maximum in-degree 2.
\item[(b)] If $D$ is a  strong digraph, then $D$ contains a spanning r-galaxy if and only if $D$ contains a strong subdigraph of even order. 
\item[(c)] If $D$ is a  strong digraph containing no spanning r-galaxy then $D-v$ contains a perfect matching for all $v \in V(D)$.
\item[(d)] We can in polynomial time decide if a strong digraph contains a spanning r-galaxy.
\item[(e)] It is NP-complete to decide, given a strong digraph and one of its arc, whether there exists a spanning r-galaxy containing (resp. avoiding) this arc.
\end{description}
\end{theorem}

We now prove our first result on the existence of pure strategy equilibria when $k=1$.
The result presents a dichotomy on the complexity of deciding whether a game possesses a pure Nash equilibrium based on the structure of the underlying digraph $D$.


\begin{theorem}
\label{thm:storng_vs_general}
When $k=1$, the problem of deciding whether there exists a pure Nash equilibrium is NP-complete.
On the other hand, the problem is polynomial-time solvable if $k=1$ and the digraph is strong.
\end{theorem}

\begin{proof}
To prove the first part of the theorem, we reduce from the problem of finding a spanning r-galaxy in a digraph, which by Theorem~\ref{DG12results}(a) is known to be NP-hard.
Let $D$ be any digraph and let $X^+ = \{x \; | \; d^+(x)=0 \}$. Now construct $D'$ from $D$ as follows: for each vertex $x \in X^+$ add two new vertices, $u_x$ and $v_x$, and the arcs $x u_x, u_x v_x, v_x x$.
We will now show that $D'$ contains a \PNa{} set if and only if $D$ contains a spanning $r$-galaxy, which will complete the proof.

Assume that $R$ is a spanning r-galaxy in $D$ and let $S=\leaves(R) \cup \{u_x \; | \; x \in X^+ \}$. 
We note that $S$ is a \PNa{} set in $D$, as $R \cup \{u_x v_x \; | \; x \in X^+ \}$ is an extension of $S$ in $D'$.
  
Now let $S$ be a \PNa{} set in $D'$. For the sake of contradiction assume that $u_x \not\in S$ for some $x \in X^+$. 
This implies that $v_x \in S$ (as otherwise $v_x \in V(D') \setminus S$ would be isolated in an extension of $S$) 
and $x \not\in S$ (as $S$ is independent in an extension of $S$).
However this implies that $u_x$ is isolated in an extension of $S$, a contradiction.  Therefore $u_x \in S$ for all $x \in X^+$.
Analogously to above, this implies that $v_x \not\in S$ and $x \not\in S$ for all $x \in X^+$.
Let $R'$ be an extension of $S$ in $D'$ and let $R$ be the subdigraph obtained from $R'$ after deleting all arcs $u_x v_x$ for all $x \in X^+$.
We now note that $R$ is a spanning r-galaxy in $D$, which completes the proof of the first part of the theorem.

Deciding whether there exists a pure strategy Nash equilibrium when $k=1$ and the digraph is strong follows immediately from Theorem~\ref{thmEquivI} and Theorem~\ref{DG12results}(d).
\end{proof}




\subsection{Pure equilibria when $k \ge 2$}

We now show that deciding whether a pure strategy Nash equilibrium exists for $k\geq 2$ is NP-hard. We emphasize that there is no qualifier along the lines of Theorem~\ref{thm:storng_vs_general} for strong digraphs as there was when $k=1$.


\begin{theorem} \label{thm:k>=2}
For all $k\ge 2$, deciding whether a strong digraph possesses a pure strategy equilibrium is NP-complete.
\end{theorem}


\begin{proof}
We will reduce from the NP-complete problem of deciding if a $3$-uniform hypergraph has a transversal of size at most $r$.
For an integer $p\ge 2$, a hypergraph $H$ is {\em $p$-uniform} if every (hyper)edge of $H$ has $p$ vertices. (A graph is a 2-uniform hypergraph.) A set $T$ of vertices in a hypergraph $H$ is a {\em traversal} if every edge of $H$ contains a vertex from $T.$ In the problem deciding if a $3$-uniform hypergraph has a transversal of size $r$, given a 3-uniform hypergraph $H$ and an integer $r$, we have to decide whether $H$ has a traversal with at most $r$ vertices. 
This problem is NP-hard, since the same problem for 2-uniform hypergraphs is NP-hard \cite{karp1972reducibility}. 

Let $H=(V,E)$ be a $3$-uniform hypergraph. We may assume that $H$ contains no isolated vertices (as such vertices can just be removed).
We will construct a strong digraph, $D^*$, such that $D^*$ contains a 
\kPNa{} set if and only if $H$ has a transversal of size $r$. This will complete the proof of the theorem.

We first construct a digraph, $D$, as follows. Let $V(D)=Z \cup E \cup V \cup X$, where $V=\{v_1,v_2,\ldots,v_n\}$ corresponds to the vertices in $H$ and
$E=\{e_1,e_2,\ldots,e_m\}$ correspond to the edges in $H$ and $Z$ and $X$ are independent sets such that $|Z|=(k-1)|E|$ and $|X|=k (|V|-r)$.
Now add arcs from every $e_i \in E$ to the three vertices in $V$ which belong to $e_i$ in $H$. Then, from every $e_i$  add $k-1$ arcs to vertices in $Z$ in such 
a way that every vertex in $Z$ gets in-degree one. Finally, add all the arcs from $V$ to $X$. See Figure~\ref{PicDH} for an illustration of $D$.

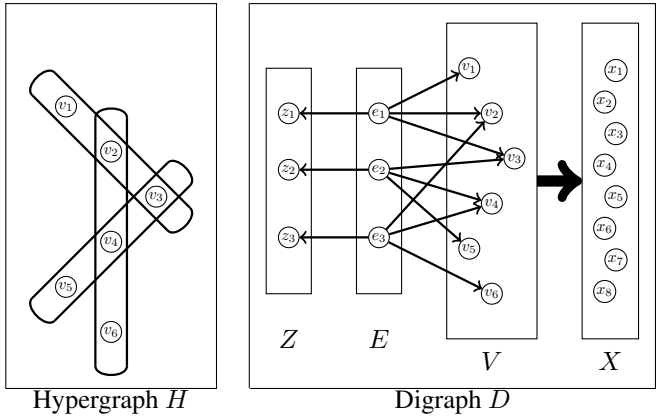
\begin{figure}[htb]
\begin{center}
\tikzstyle{vertexWC}=[circle, draw, minimum size=6pt, scale=0.6, inner sep=0.8pt]
\begin{tabular}{|c|c|c|} \cline{1-1} \cline{3-3}
\begin{tikzpicture}[scale=0.3]
  \node (v1) at (0,12) [vertexWC]{$v_1$};
  \node (v2) at (2,10) [vertexWC]{$v_2$};
  \node (v3) at (4,8) [vertexWC]{$v_3$};
  \node (v4) at (2,6) [vertexWC]{$v_4$};
  \node (v5) at (0,4) [vertexWC]{$v_5$};
  \node (v6) at (2,2) [vertexWC]{$v_6$};

  \draw[line width=0.03cm] (-0.5,13.5) -- (5.5,7.5);
  \draw[line width=0.03cm] (-1.5,12.5) -- (4.5,6.5);
  \draw[line width=0.03cm] (-1.5,12.5) to [out=135, in=135] (-0.5,13.5);
  \draw[line width=0.03cm] (5.5,7.5) to [out=315, in=315] (4.5,6.5);

  \draw[line width=0.03cm] (4.5,9.5) -- (-1.5,3.5);
  \draw[line width=0.03cm] (5.5,8.5) -- (-0.5,2.5);
  \draw[line width=0.03cm] (-1.5,3.5) to [out=225, in=225] (-0.5,2.5);
  \draw[line width=0.03cm] (4.5,9.5) to [out=45, in=45] (5.5,8.5);

  \draw[line width=0.03cm] (1.3,11.5) -- (1.3,0.5);
  \draw[line width=0.03cm] (2.7,11.5) -- (2.7,0.5);
  \draw[line width=0.03cm] (1.3,11.5) to [out=90, in=90] (2.7,11.5);
  \draw[line width=0.03cm] (1.3,0.5) to [out=270, in=270] (2.7,0.5);
\end{tikzpicture} & & 
\begin{tikzpicture}[scale=0.3]
  \node at (3,14.5) { };
  \node (v1) at (1,12) [vertexWC]{$v_1$};
  \node (v2) at (2,10) [vertexWC]{$v_2$};
  \node (v3) at (3,8) [vertexWC]{$v_3$};
  \node (v4) at (2,6) [vertexWC]{$v_4$};
  \node (v5) at (1,4) [vertexWC]{$v_5$};
  \node (v6) at (2,2) [vertexWC]{$v_6$};

  \draw [] (0,0) rectangle (4,14);   \node at (2,-1) {$V$};

  \node (e1) at (-3,10) [vertexWC]{$e_1$};
  \node (e2) at (-3,7.5) [vertexWC]{$e_2$};
  \node (e3) at (-3,4.5) [vertexWC]{$e_3$};

  \draw [] (-4,2) rectangle (-2,12);   \node at (-3,0) {$E$};

  \node (z1) at (-7,10) [vertexWC]{$z_1$};
  \node (z2) at (-7,7.5) [vertexWC]{$z_2$};
  \node (z3) at (-7,4.5) [vertexWC]{$z_3$};

  \draw [] (-8,2) rectangle (-6,12);   \node at (-7,0) {$Z$};

  \node (x1) at (7.5,11.9) [vertexWC]{$x_1$};
  \node (x2) at (7,10.5) [vertexWC]{$x_2$};
  \node (x3) at (7.5,9.1) [vertexWC]{$x_3$};
  \node (x4) at (7,7.7) [vertexWC]{$x_4$};
  \node (x5) at (7.5,6.3) [vertexWC]{$x_5$};
  \node (x6) at (7,4.9) [vertexWC]{$x_6$};
  \node (x7) at (7.5,3.5) [vertexWC]{$x_7$};
  \node (x8) at (7,2.1) [vertexWC]{$x_8$};

  \draw [] (6,0) rectangle (8.5,14);   \node at (7.25,-1) {$X$};

  \draw[->, line width=0.03cm] (e1) -- (z1);
  \draw[->, line width=0.03cm] (e2) -- (z2);
  \draw[->, line width=0.03cm] (e3) -- (z3);

  \draw[->, line width=0.03cm] (e1) -- (v1);
  \draw[->, line width=0.03cm] (e1) -- (v2);
  \draw[->, line width=0.03cm] (e1) -- (v3);

  \draw[->, line width=0.03cm] (e2) -- (v3);
  \draw[->, line width=0.03cm] (e2) -- (v4);
  \draw[->, line width=0.03cm] (e2) -- (v5);

  \draw[->, line width=0.03cm] (e3) -- (v2);
  \draw[->, line width=0.03cm] (e3) -- (v4);
  \draw[->, line width=0.03cm] (e3) -- (v6);

  \draw[->, line width=0.15cm] (4,7) -- (6,7);
\end{tikzpicture}  \\ \cline{1-1} \cline{3-3}
\multicolumn{1}{c}{Hypergraph $H$} & \multicolumn{1}{c}{ } & \multicolumn{1}{c}{Digraph $D$} \\
\end{tabular}
\caption{The digraph $D$ constructed from $H$ in the proof of Theorem~\ref{thm:k>=2} when $k=2$ and $r=2$. Note that $|Z|=(k-1)\cdot|E|=1 \cdot 3$ and $|X|=k\cdot(|V|-r)= 2\cdot(6-2)=8$.}
\label{PicDH} 
\end{center} 
\end{figure}

We now construct $D^*$ from $D$. For every vertex, $q \in Z \cup X$, add the vertices $\{ u_q,w_q,s_q,t_q\}$ and the arcs $\{q u_q, u_q w_q,w_q q, qs_q, s_q t_q\}$ 
and all arcs from $t_q$ to $E$. See Figure~\ref{PicGadget} for an illustration of this gadget (ignore the different colors of the vertices for now).
Note that $D^*$ is strong as every vertex in $D$ lies on an $(E,X \cup Z)$-path and the gadgets imply the existence of paths from every vertex in $X \cup Z$ to 
every vertex in $E$.

\begin{figure}[htb]
\begin{center}
\tikzstyle{vertexBC}=[circle, draw, minimum size=6pt, scale=0.8, inner sep=0.8pt]
\tikzstyle{vertexWC}=[circle, draw, minimum size=6pt, scale=0.6, inner sep=0.8pt]
\tikzstyle{vertexWD}=[circle, draw, top color=gray!40, bottom color=gray!10, minimum size=6pt, scale=0.6, inner sep=0.8pt]
\begin{tabular}{|c|c|} \hline
\begin{tikzpicture}[scale=0.3]
  \node at (1,4.5) { };  
  \node (q) at (3,1) [vertexBC]{$q$};
  \draw [] (11,1) rectangle (13,4);   \node at (12,0) {$E$};
\end{tikzpicture}  & 
\begin{tikzpicture}[scale=0.3]
  \node at (5,4.5) { };    
  \node (q) at (3,1) [vertexBC]{$q$};
  \node (u) at (5,3) [vertexWD]{$u_q$};
  \node (w) at (1,3) [vertexWC]{$w_q$};
  \node (s) at (6,1) [vertexWD]{$s_q$};
  \node (t) at (8,3) [vertexWC]{$t_q$};
  \draw [] (11,1) rectangle (13,4);   \node at (12,0) {$E$}; 
  \draw[->, line width=0.03cm] (q) -- (u);
  \draw[->, line width=0.03cm] (u) -- (w);
  \draw[->, line width=0.03cm] (w) -- (q);
  \draw[->, line width=0.03cm] (q) -- (s);
  \draw[->, line width=0.03cm] (s) -- (t);
  \draw[->, line width=0.15cm] (t) -- (11,2.5);
\end{tikzpicture}  \\ \hline
\end{tabular}
\caption{The gadgets added in order to transform the digraph $D$ into $D^*$.}
\label{PicGadget} 
\end{center} 
\end{figure}
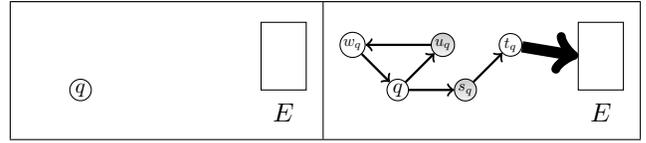

  First assume that $H$ has a transversal, $T$, of size $r$. Let every vertex in $E$ (in $D^*$) buy and choose its $k-1$ out-neighbours in $Z$ and one of 
the vertices in $T$ (which exists as every edge in $H$ contains a vertex from $T$). Let every vertex in $V \setminus T$ buy and chooses $k$ out-neighbours 
in $X$ such that every vertex in $X$ is chosen (which is possible as $|X|=k(|V|-r)$). Finally, for every $q \in Z \cup X$, we let $u_q$ and $s_q$ buy (that is,
the gray vertices in Figure~\ref{PicGadget}) and chooses $w_q$ and $t_q$, respectively.  We note that this buyer assignment is a \kPNa{} set in $D^*$.

  Conversely, assume that $S$ is a \kPNa{} set in $D^*$ (that is, $S$ is the set of buyers in a pure strategy Nash equilibrium). 
We will show that $H$ contains a transversal of size $r$, which will complete the proof. 
Let $q \in Z \cup X$ be arbitrary. If $q \in S$, then $u_q \not\in S$ (as $d_{D^*}^+(q)=2$ and $q u_q \in A(D^*)$). Therefore, $w_q \in S$ (as no in-neighbour of $w_q$
buys), which contradicts the fact that $q \in S$, as $N_{D^*}^+(w_q)=\{q\}$. So no vertex from $Z \cup X$ belongs to $S$. 

As no vertex from $Z \cup X$ belongs to $S$, we note that $\{u_q,s_q\} \subseteq S$ and $\{w_q,t_q\} \cap S = \emptyset$ for every $q \in Z \cup X$  (that is, 
the gray vertices in Figure~\ref{PicGadget} belong to $S$). 
This implies that $E \subseteq S$ (as no in-neighbour of any vertex in $E$ buys) and for every $e_i \in E$ it chooses $k-1$ vertices in $Z$ and one vertex in $V$, 
which implies that this vertex does not belong to $S$. So $V \setminus S$ is a transversal in $H$. 

As every vertex in $X$ must be chosen by at least one buyer, we note that at least $|V|-r$ vertices in $V$ buy (as $|X|=k(|V|-r)$. So $V \setminus S$ contains at most $r$ vertices and these
vertices form a transversal in $H$. So $H$ contains a transversal of size at most $r$ and therefore also a transversal of size exactly $r$ (just add arbitrary vertices to the transversal).  
That is, $D^*$ contains a \kPNa{} set if and only if $H$ has a transversal of size $r$, as desired.
\end{proof}

\subsection{Sufficient conditions for pure Nash equilibria}

We now present some structural properties of $D$ for which a pure strategy Nash equilibrium is guaranteed to exist.
We also show the complexity of determining whether the property is satisfied.
Before stating the result, note that we write $\Delta^+(D)$ to denote the maximum out-degree, taken over all vertices, in digraph $D$.
Also, a strong component $C$ of digraph $D$ is {\em terminal} if there is no arc from a vertex in $C$ to a vertex outside of $C$.

\iflong
\begin{theorem}
\label{thm:structural_results}
\fi
\ifshort
\begin{theorem}[$\star$]
\label{thm:structural_results}
\fi
Let $D$ be a digraph. Then, a pure Nash equilibrium exists if one of the following holds:
\begin{enumerate}[label=(\roman*)]
    \item\label{sufficient1}
    $D$ contains no odd cycles and $\Delta^+(D) \leq 1$.
    \item\label{sufficient2}
    Every terminal strong component of $D$ either contains an even cycle or it is a single vertex.
    \item $D$ is acyclic.
    \item $D$ is bipartite.
\end{enumerate}
\end{theorem}

\iflong
\begin{proof}
    \begin{enumerate}[label=(\roman*)]
    \item 
    We will prove the theorem by induction on the number of vertices, $n$, in $D$. Clearly the theorem is true if $n$ equals 1 or 2. 
Let $D$ be a digraph containing no odd cycles and with $\Delta^+(D) \leq 1$.
If $\delta^-(D) \geq 1$ then $D$ consists of a collection of vertex-disjoint even cycles. In this case $D$ contains a \IPNa{} set by Theorem~\ref{thmEquivI} and Theorem~\ref{DG12results}(b).
So we may assume that there exists a vertex $x \in V(D)$ with $d^-(x)=0$.
If $d^+(x)=0$ then we are done by using our induction hypothesis on $D-x$ (and adding $x$ to the \PNa{} set obtained in $D-x$). So we may assume that 
$d^+(x)=1$. Let $xy$ be the unique arc leaving $x$ in $D$.

  Let $D' = D - \{x,y\}$. Clearly $D'$ contains no odd cycles and $\Delta^+(D') \leq 1$.  So by 
induction there exists a \PNa{} set, $S'$, in $D'$. Let $S = S' \cup \{x\}$.
We will show that $S$ is a \PNa{} set in $D$. Let $R'$ be an extension of $S'$ in $D'$. If any vertex $s' \in S'$ has out-degree zero in $D'$ but had an arc
into $y$ in $D$ then we add the arc $s'y$ to $R'$ and finally we add the arc $xy$ to $R'$ and call the resulting subdigraph $R$. We note that $R$ is an extension 
of $S$ in $D$, thereby showing that $S$ is a \PNa{} set in $D$, as desired.

    \item
    Let $D$ be a digraph and let $C_1,C_2,\ldots, C_r$ denote the terminal strong components of $D$. Assume that each $C_i$ is either a single vertex or
contains an even cycle. 
Start by letting $R$ be empty and then for each $i=1,2,\ldots,r$ add an even cycle from $C_i$ to $R$ if $C_i$ is not a single vertex and add $C_i$ to $R$ if $C_i$ is a single vertex. 
Note that every vertex in $V(D) \setminus V(R)$ has a path to $R$. 
So if $V(R) \not= V(D)$ then add an arc from $V(D) \setminus V(R)$ to $V(R)$ to $R$ (thereby also adding a vertex to $V(R)$). 
Continue this process until $V(R) = V(D)$.  
We now note that every vertex has out-degree one in $R$ except the vertices that belonged to $C_i$'s of order one.  
And $R$ contains no odd cycles.  
The result now follows from part~\ref{sufficient1}.

    \item
  We will now show that every acyclic digraph contains a \IPNa{}. We will show this by induction on the order of the digraph.  Clearly this is true if the 
order is 1 or 2.  Let $D$ be an acyclic digraph of order $n$. As $D$ is acyclic it must contain a vertex, $x$, with indegree 0.
If $x$ is isolated in $D$ then use induction on $D-x$ and add $x$ to the \PNa{} set obtained in $D-x$ in order to obtain a \PNa{} set in $D$.
So we may assume $x$ is not isolated and let  $D'$ be obtained from 
$D$ by removing $x$ and one of its out-neighbours, $y$. By induction there exists a \PNa{} set, $S'$, in $D'$. 
We will show that $S' \cup \{x\}$ is a \PNa{} set in $D$. Let $R'$ be an extension of $S'$ in $D'$. If any vertex $s' \in S'$ has out-degree zero in $D'$ but had an arc
into $y$ in $D$ then we add the arc $s'y$ to $R'$ and finally we add the arc $xy$ to $R'$ and call the resulting subdigraph $R$. We note that $R$ is an extension 
of $S$ in $D$, thereby showing that $S$ is a \PNa{} set in $D$, as desired.

    \item
   That every bipartite digraph contains a \IPNa{} follows from part~\ref{sufficient2}.
\end{enumerate}
\end{proof}
\fi

Before proceeding, we remark on the complexity of determining the properties above.
First, observe from parts (iii) and (iv) of Theorem~\ref{thm:structural_results} that deciding if a digraph has a spanning r-galaxy and deciding if it has a pure Nash set are fundamentally different problems (even though they are the same in many cases as per Theorem~\ref{thmEquivI}).
This holds for part (iii) because it follows from Theorem~\ref{DG12results}(a) that is is NP-hard to decide if an acyclic digraph contains a spanning $r$-galaxy.
Similarly, for part (iv) since, by  Theorem~\ref{DG12results}(a), it is NP-hard to decide if a bipartite digraph contains a spanning r-galaxy.

\section{The complexity of mixed Nash equilibria}

In this section, we turn to the issue of computing mixed strategy Nash equilibria. 

\subsection{Mixed equilibria when $k =1$}
First we prove how to efficiently find a mixed Nash equilibrium when $k=1$.

\begin{theorem}
    \label{thm:k_one_poly}
    When $k=1$, we can compute a mixed Nash equilibrium in linear time.
\end{theorem}
\begin{proof}
    The algorithm works in three phases. In the first phase, we create a subgraph $D'$ of digraph $D$ as follows. We arbitrarily order the vertices of the digraph, and for each one we arbitrarily pick an out-neighbor, if it has an out-neighbor. The idea is that if a player buys, then, they will choose the identified out-neighbor with probability 1.
    In the second phase, we focus on $D'$ and more specifically on vertices with in-degree zero (called sources), i.e., players that have to buy in every Nash equilibrium. So, while there is a source in the remaining directed graph, we arbitrarily pick one of them and the chosen out-neighbor, if it exists.
    We set the strategy of the source to buy, the strategy of the chosen out-neighbor (if any) to not buy, we delete these vertices from the digraph, and we check again for sources in the remaining graph.
    At the end of the second phase, the remainder of digraph $D'$  consists of disjoint directed cycles {\em only}. 
    For each remaining vertex, we set the probability of buying to $1-c$ and choose the designated out-neighbor.

    We argue that the constructed strategy profile is a mixed Nash equilibrium. There are three types of players depending on the strategy they play and we will consider them independently and prove that each of them plays a best response. 

    \begin{itemize}
        \item Players that buy with probability 1 incur cost $c$. Observe that the constructed strategy guarantees that for every such player there is no in-neighbor that buys and chooses her. Hence, if they deviates to not buy, they will incur cost of $1$. Hence, every such player plays a best response.
        \item Every player that buys with probability 0, must have an in-neighbor that buys with probability 1 and chooses them. Hence, every such player incurs cost 0, which is clearly a best response.
        \item For the players who mix between two actions we will calculate the cost of each action and we will prove that the expected cost in each case is optimal. Indeed, observe that the expected cost for buying is $c$; the cost is the same for every out-neighbor the player chooses. 
        Furthermore, since by construction of the strategy profile exactly one in-neighbor of the player will choose the player, the expected cost for not buying is calculated as follows. With probability $p$ the in-neighbor will buy and will pick the agent and the agent will incur cost 0. 
        With probability $1-p = 1-(1-c)=c$ the in-neighbor will not buy and the player will incur cost $1$. Hence, the expected cost for not buying is again $c$. So, the player plays with positive probability only actions that yield minimum cost and thus they have no incentives to deviate.
    \end{itemize}
Hence, the constructed strategy profile is a Nash equilibrium.
\end{proof}

\subsection{Mixed equilibria when $k \ge 2$}

Unfortunately\iflong for us, as we will argue next\fi, for every $k \geq 2$, \iflong the problem of finding a mixed Nash equilibrium becomes hopelessly intractable.
    Formally, \fi the problem is PPAD-hard even for finding an approximate Nash equilibrium, i.e., there is unlikely to exist a polynomial-time algorithm (unless PPAD=P \cite{PAPADIMITRIOU1994498}) that finds a strategy profile that every player cannot decrease their cost by ``much''.
The result essentially follows from the construction of~\citet{DODINH2024106486}, where they study public good games {\em without} the sharing constraint $k$. In their reduction, they construct a directed graph with maximum out-degree 2, hence their PPAD-hardness holds for our model, for any $k \geq 2$. 
We state the result for completeness.
\begin{theorem}[\citet{DODINH2024106486}]
    For every $k \geq 2$, finding a mixed Nash equilibrium is PPAD-hard.
\end{theorem}

\section{The efficiency of Nash equilibria}\label{sec:efficiency}

We now consider the efficiency of equilibria.
As before, we begin with restricting attention to pure strategies before allowing for randomisation. 
By definition, mixing can only improve welfare since every pure strategy can be viewed as a special case of a mixed one.

\subsection{Efficiency of pure strategy equilibria}




As with deciding on the existence of pure strategy Nash equilibria, there is a difference between the case of $k=1$ and that of $k \ge 2$.
We begin with a general result that holds for all $k$.

\iflong
\begin{theorem}\label{thm:PoS}
\fi
\ifshort
\begin{theorem}[$\star$]\label{thm:PoS}
\fi
Fix $k \ge 1$. Then 
$k\le {\rm PoS}_k \le k + \frac{1}{k+1}$.
\end{theorem}

\iflong
\begin{proof}
\fi
\ifshort
\begin{proof}[Proof sketch]
The proof of the upper bound is rather technical and it is deferred to the supplementary material.
\fi
\iflong
First, we will prove the upper bound on ${\rm PoS}(D).$
Let $D$ be a digraph with $n$ vertices and let $S$ be a \kPNa{} set in $D$ with $|S|=pn_k(D)$.
Let $X = V(D) \setminus S$. Note that $X$ are the non-buyers and $S$ are the buyers in the pure Nash equilibrium.
The following two claims complete the proof of the theorem.

\2

{\bf Claim 1:} {\em If $\frac{|X|}{|S|}\geq \frac{k}{k^2+k+1}$ then $\frac{pn_k(D)}{b_k(D)} \leq k + \frac{1}{k+1}$.}


\2

{\bf Proof of Claim 1.} As no buyer can choose more than $k$ people we note that  $b_k(D) \geq n/(k+1)$. Therefore, the following holds.

\[
\begin{array}{rcl} \2
\frac{pn_k(D)}{b_k(D)} & \leq & \frac{|S|}{ n/(k+1) } \\ \2
     & =    &  \frac{(k+1)|S|}{|S|+|X|}  \\ \2
     & =    & \frac{k+1}{1+|X|/|S|} \\ \2
     & \leq &  \frac{k+1}{1+k/(k^2+k+1)} \\ \2
     & = & \frac{(k+1)(k^2+k+1)}{k^2+2k+1} \\ \2
     & = & \frac{k^2+k+1}{k+1} \\
     & = &  k + \frac{1}{k+1} \\
\end{array}
\]

\2

{\bf Claim 2:} {\em If $\frac{|X|}{|S|}< \frac{k}{k^2+k+1}$, then $\frac{pn_k(D)}{b_k(D)} \leq k + \frac{1}{k+1}$.}

\2

{\bf Proof of Claim 2.} We first consider the spanning subdigraph $D'$ of $D$ defined as follows. We let $V(D')=V(D)$ and for every vertex
$s \in S$, we keep the $\max\{k,d_D^+(s)\}$ arcs out of $s$ to those nominated by $s$.
All other arcs of $D$ are removed.
That is, $D'$ is obtained from $D$ by removing all arcs $uv$ where $u$ does not choose $v$. Note that in $D'$ all arcs are $(S,X)$-arcs.

For each $x \in X$ pick a vertex $y \in N_{D'}^-(x)$ arbitrarily and let $Y$ be the set of these vertices. Note that $|Y| \leq |X|$ and
$N_{D'}^+(Y) = X$. Let $Y_1$ contain all the vertices $y \in Y$ for which all neighbors of $y$ in $D'$ have in-degree 1 in $D'$.
That is, $y$ belongs to $Y_1$ if and only if everyone who is chosen by $y$ is not chosen by anyone else.
Let $Y_2 = Y \setminus Y_1$ and let $Z = S \setminus Y$.  
Also, let $X_1 = N_{D'}^+(Y_1)$ and $X_2= X \setminus X_1$. We now prove the following subclaims, where Subclaim~2.E completes the proof.

\2

{\bf Subclaim 2.A:} {\em $Z$ is an independent set in $D$.}

\2

{\bf Proof of Subclaim 2.A} For the sake of contradiction, assume that $Z$ is not independent in $D$ and $z_1 z_2$ is an arc of $D$ with 
both endpoints in $Z$.
If $z_1$ now chooses $z_2$ instead of one of the other people to whom $z_1$ is currently chooses, then $z_2$ does not need to buy, but we still have a Nash equilibrium.
That is, $S \setminus \{z_2\}$ is a \kPNa{} set in $D$, contradicting the minimality of $|S|$.

\2

{\bf Subclaim 2.B:} {\em The set of buyers in $X_1 \cup Y_1$ can choose at most $k \cdot |X_1|$ people in $Z$ in any buyer assignment.}

\2

{\bf Proof of Subclaim 2.B} Let $y \in Y_1$ be arbitrary and let $X_1^y = N_{D'}^+(y)$. 
If $y$ has no arc in $Z$ in $D$ then we note that $\{y\} \cup X_1^y$ can choose at most $k \cdot |X_1^y|$ buyers in $Z$. 

So now assume that there is an arc from $y$ to $z$ for some $z \in Z$ in $D$. If any vertex $w \in X_1^y$ has $k$ out-neighbours, $W$, in $Z$ then 
$S \cup \{w\} \setminus (W \cup \{z\})$ is a \kPNa{} set in $D$ (where $w$ chooses all people in $W$ and $y$ chooses to $z$ instead of to $w$),
contradicting the minimality of $|S|$. Therefore each vertex in $X_1^y$ can chooses at most $k-1$ people of $Z$. As $y$ is not choosing $z$
we note that the out-degree of $y$ is greater than $k$ in $D$ and therefore $|X_1^y|=k$. So $\{y\} \cup X_1^y$ can choose 
at most $k + (k-1)|X_1^y| = k \cdot |X_1^y|$ people in $Z$. 

So as in all cases, $\{y\} \cup X_1^y$ can choose at most $k \cdot |X_1^y|$ buyers in $Z$, we note that by summing over all
$y \in Y_1$ we get that the set of buyers in $X_1 \cup Y_1$ can choose at most $k \cdot |X_1|$ people in $Z$ in any buyer assignment, 
as desired.

\2

{\bf Subclaim 2.C:} {\em For every $y \in Y_2$ we have $|N_D^+(y) \cap Z | \leq 1$.}

\2

{\bf Proof of Subclaim 2.C} Assume for the sake of contradiction that $y \in Y_2$ and $|N_D^+(y) \cap Z | \geq 2$.
Let $\{z_1,z_2\} \subseteq N_D^+(y) \cap Z$ be arbitrary. As $y \not\in Y_1$ we note that some $x \in N_{D'}^+(y)$ has indegree at least two in $D'$.
Without loss of generality we may assume that either $z_1$ has no arc into $x$ in $D'$ or both $z_1$ and $z_2$ have arcs into $x$ in $D'$.
If $y$ now chooses $z_1$ instead of $x$, then $z_1$ does not need to buy, but we still have a Nash equilibrium. 
That is, $S \setminus \{z_1\}$ is a \kPNa{} set in $D$, contradicting the minimality of $|S|$.

\2

{\bf Subclaim 2.D:} {\em $b_k(D) \geq |X| + |Z|-k|X|$.}

\2

{\bf Proof of Subclaim 2.D} We consider any buyer assignment in $D$. 
By Subclaim~2.B we note that the set of buyers in $X_1 \cup Y_1$ can choose at most $k \cdot |X_1|$ people in $Z$ in this buyer assignment.
It is also clear that the set of buyers in $X_2$ can choose at most $k \cdot |X_2|$ players in $Z$. 
Therefore $X \cup Y_1$  can choose at most $k \cdot |X|$ people in $Z$.

By Subclaim~2.A and Subclaim~2.C we note that no two people in $Z$ can chosen by the same player unless the buyer is
in $X \cup Y_1$.
Let $Z^*$ denote all people in $Z$ who have not been chosen
by anybody in $X \cup Y_1$. So all people in $Z^*$ must be chosen by different players.

By the above the minimum number of buyers needed is at least $|X|$ buyers from $X \cup Y_1$ which choose $k \cdot |X|$ people in $Z$.
And an additional $|Z^*| \geq |Z| - k \cdot |X|$ offers for the players in $Z^*$.
Therefore $b_k(D) \geq |X| + (|Z|-k|X|)$, as desired.

\2

{\bf Subclaim 2.E:} {\em $\frac{pn_k(D)}{b_k(D)} \leq k + \frac{1}{k+1}$.}

\2

{\bf Proof of Subclaim 2.E} By the assumption in Claim~2 and Subclaim 2.D, we have the following:

\[
\begin{array}{rcl} \2
\frac{pn_k(D)}{b_k(D)} & \leq & \frac{|S|}{ |X| + |Z|-k|X| } \\ \2
     & =    & \frac{|S|}{|X| + (|S|-|Y|)-k|X|}  \\ \2
     & \le & \frac{|S|}{|S|-k|X|}  \\ \2
     & \leq & \frac{1}{1-k|X|/|S|}  \\ \2
    & < & \frac{1}{1-k^2/(k^2+k+1)} \\ \2
     & \leq & \frac{k^2+k+1}{k+1} \\ 
     & = &  k + \frac{1}{k+1} \\
\end{array}
\]


\fi 

Next we show that the price of stability is at least $k$.
Figure~\ref{Pic1} presents an out-tree where every non-leaf has out-degree $k=2$ and $r=3$ levels.
Consider a generalization of this example where every vertex has out-degree $k$ and there are $r$ levels, where $r$ is odd.
Comparing efficient equilibria and efficient outcomes in such structures parallels our analysis of Figure~\ref{Pic1}.
The most efficient buyer set involves forcing all vertices in the first two levels to buy and then alternating levels with non-buyers and buyers thereafter.
But this is not an equilibrium because in the unique equilibrium the root must buy and therefore those in the second level will not. This continues, meaning that each vertex in every odd-numbered level buys.
As $r$ increases in any such structure, the price of stability tends to $k$. 
\end{proof}

In fact, we believe that $k$ is the right answer for PoS.

\begin{conjecture}\label{conj1}
    If a digraph admits a pure strategy Nash equilibrium, then the price of stability is equal to $k$.
\end{conjecture}

The following theorem verifies our conjecture for $k=1$.
Before stating the result, we introduce the following notation:\ given a digraph $D$, let $\alpha'(D)$ denote the size of a maximum matching in $D$.
\iflong
\begin{theorem} \label{priceOfStabilityI}    
\fi
\ifshort
\begin{theorem}[$\star$] \label{priceOfStabilityI}
\fi
Fix $k=1$. If the model admits a pure strategy Nash equilibrium then the price of stability is equal to 1.
Moreover, the most efficient outcomes (equilibrium and non-equilibrium alike) have exactly $n - \alpha'(D)$ buyers.
\end{theorem} 

\iflong
\begin{proof}
\fi
\ifshort
\begin{proof}[Proof sketch]
\fi
Assume that the digraph $D$ contains a \IPNa{} set. We first show that $n-\alpha'(D)=b_k(D)$.
Let $M$ be a maximum matching in $D$ and let $S$ denote all the heads of the arcs in $M$. Now $V(D) \setminus S$ is a \IDr{} set
in $D$ of size $n-|S| = n-\alpha'(D)$ (where each vertex in $S$ gets chosen by its in-neighbour in $M$), implying that $b_1(D) \leq n-\alpha'(D)$.

Conversely, let $R$ be an extension of a \IDr{} set, $Q$, where $|Q|=b_1(D)$. For each vertex $x \in V(D) \setminus Q$ we can pick an arc into $x$ from $R$.
Note, these $n-|Q|$ arcs form a matching, implying that $\alpha'(D) \geq n-b_1(D)$. As $b_1(D) \leq n-\alpha'(D)$, we obtain $b_1(D) = n-\alpha'(D)$.

We will now prove that $pn_1(D)=n-\alpha'(D)$. First let  $H$ be an extension of a \IPNa{} set, $S$, where $|S|=pn_1(D)$. 
For each vertex $x \in V(D) \setminus S$ we can pick an arc into $x$ from $H$.
Note that these $n-|S|$ arcs form a matching, implying that $\alpha'(D) \geq n-pn_1(D)$. 
Or, alternatively, that $pn_1(D) \geq n - \alpha'(D)$. 

Next we now prove that $pn_1(D) \leq n - \alpha'(D)$, which will complete the proof of the theorem.
For the sake of contradiction assume that $pn_1(D) > n - \alpha'(D)$.
Let $M$ be a maximum matching in $D$ and let $H$ be an extension of a \IPNa{} set, $S$, where $|S|=pn_1(D)$. 
Let $X_1$ be the vertices in $H$ with in-degree one and let $X_2$ be the vertices in $H$ with in-degree at least 2.
Let $Y_1$ be the vertices in $H$ with an arc into $X_1$ and let $Y_2$ be the vertices in $H$ with an arc into $X_2$ 
and let $Z$ be the isolated vertices in $H$ (see Figure~\ref{Pic2}).

\begin{figure}[htb]
\begin{center}
\tikzstyle{vertexWC}=[circle, draw, minimum size=10pt, scale=0.75, inner sep=0.8pt]
\begin{tikzpicture}[scale=0.24]

  \node (z1) at (-9,1) [vertexWC]{}; 
  \node (z2) at (-7,1) [vertexWC]{}; 
  \node (z3) at (-5,1) [vertexWC]{}; 
  \draw [] (-10,0) rectangle (-4,2);   \node at (-10.8,1) {$Z$};

  \node (x1) at (2,8) [vertexWC]{};
  \node (y1) at (1,1) [vertexWC]{}; \draw[->, line width=0.03cm] (y1) -- (x1);
  \node (y2) at (3,1) [vertexWC]{}; \draw[->, line width=0.03cm] (y2) -- (x1);

  \node (x2) at (7,8) [vertexWC]{};
  \node (y3) at (5,1) [vertexWC]{}; \draw[->, line width=0.03cm] (y3) -- (x2);
  \node (y4) at (7,1) [vertexWC]{}; \draw[->, line width=0.03cm] (y4) -- (x2);
  \node (y5) at (9,1) [vertexWC]{}; \draw[->, line width=0.03cm] (y5) -- (x2);

  \node (x3) at (12,8) [vertexWC]{};
  \node (y6) at (11,1) [vertexWC]{}; \draw[->, line width=0.03cm] (y6) -- (x3);
  \node (y7) at (13,1) [vertexWC]{}; \draw[->, line width=0.03cm] (y7) -- (x3);

  \node (x4) at (17,8) [vertexWC]{};
  \node (y8) at (17,1) [vertexWC]{}; \draw[->, line width=0.03cm] (y8) -- (x4);

  \node (x5) at (20,8) [vertexWC]{};
  \node (y9) at (20,1) [vertexWC]{}; \draw[->, line width=0.03cm] (y9) -- (x5);

\draw [] (0,0) rectangle (14,2);   \node at (-1,1) {$Y_2$};
\draw [] (0,7) rectangle (14,9);   \node at (-1,8) {$X_2$};
\draw [] (16,0) rectangle (21,2);  \node at (22,1) {$Y_1$};
\draw [] (16,7) rectangle (21,9);  \node at (22,8) {$X_1$};
\end{tikzpicture} 
\caption{An illustration of $Z$, $X_1$, $X_2$, $Y_1$ and $Y_2$ in the proof of Theorem~\ref{priceOfStabilityI}, where the digraph $H$ is shown above.
Note that $S=Z \cup Y_1 \cup Y_2$.}
\label{Pic2}
\end{center}
\end{figure}
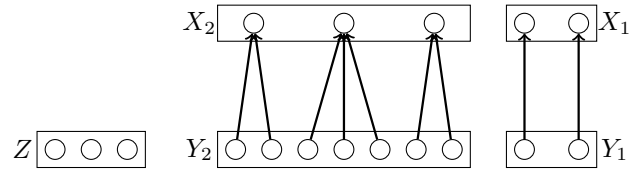

  We first consider the $2$-edge-colored graph, $G$, defined as follows. Let $V(G)=V(D)$ and let the edges in the underlying undirected graph of $H$ have color $1$ in $G$.
Let the underlying undirected edges of the matching, $M$, have color 2 in $G$. This defines the $2$-edge-colored multi-graph $G$ (there could be both an edge of color 1 and 
an edge of color 2 between 2 vertices).

We then define the graph $G'$ as follows. Let $V(G') = Z \cup X_2 \cup Y_2$ and define the edge-set of $G'$, such that $ab \in E(G')$ if and only if there 
is an alternating path between $a$ and $b$ in $G$ starting and ending with edges of color 2 and such that $a$ and $b$ are the only vertices on the path 
from $Z \cup X_2 \cup Y_2$. We need the following claim.

\2

{\bf Claim A\ifshort $(\star)$ \fi:} $|E(G')| > |X_2|$.

\2
\iflong
{\em Proof of Claim A:}
Recall that we, for the sake of contradiction, assumed that $ pn_1(D) > n - \alpha'(D)$, which implies that the following holds
\begin{align*}
    |Y_1|+|Y_2| + |Z| = |S| = pn_1(D) > n - \alpha'(D) = n - |M| 
\end{align*}
which in turn implies that
\begin{align}
\label{eq:aux_A}
2 |Y_2| > 2n - 2|M| - 2|Y_1| - 2|Z| 
\end{align}


Let $R_M$ denote all vertices in $V(D)$ that are not saturated by the matching $M$.
Let $r = |R_M| = n - 2|M|$.
Let $u \in V(G')$ ($= Z \cup X_2 \cup Y_2$) be arbitrary. 
We note that we may start an alternating path in $u$ (using the colors in $G'$) and continue until we either return to $V(G')$ or end in a vertex in $R_M$.
If we return to $V(G')$ then this will give rise to an edge in $G'$ (incident with $u$).
If we end in a vertex $v \in R_M$ ($u=v$ is possible) then we note that if we had picked a different vertex than $u$ to start from then we can not end in $v$.
So the $r$ vertices in $R_M$ can cause at most $r$ vertices in $V(G')$ to have degree zero. So at least $|V(G')| - r$ vertices in $G'$ have degree at
least one (in fact degree equal to one as it is not difficult to see that $\Delta(G') \leq 1$). 

Therefore, the following holds, by \eqref{eq:aux_A} above and the fact that $|X_1|=|Y_1|$.

\[
\begin{array}{rcl} \2
|E(G')| & \geq & \frac{|V(G')| - r}{2} \\ \2
 & = & \frac{|Z| + |X_2|+|Y_2| - (n - 2|M|)}{2} \\ \2
 & > & \frac{|Z| + |X_2|+ (2n - 2|M| - 2|Y_1| - 2|Z|) - |Y_2| - (n - 2|M|)}{2} \\ \2
 & = & \frac{|X_2|+ n - 2|Y_1|  - |Z| - |Y_2|}{2} \\ \2
 & = & \frac{|X_2| + |Z| + |X_1| + |X_2| + |Y_1| + |Y_2|  - 2|Y_1| - |Z| - |Y_2|}{2} \\ \2
 & = & \frac{|X_2|+ |X_2|}{2} \\ \2
 & = & |X_2| \\ 
\end{array} 
\]

 This completes the proof of Claim~A.

\2
\fi

We now return to the proof of Theorem~\ref{priceOfStabilityI}.
Let $G''$ be obtained from $G'$ by deleting at most $|X_2|$ edges in the following way.  For each $x \in X_2$ we delete at most one edge using the below 
approach.

\begin{itemize}
 \item If $d_{H}^-(x) = 2$ and there is an edge in $G'$ between the two vertices in $N_H^-(x)$ then delete this edge from $G'$.
 \item If we did not delete an edge above, then delete the edge incident with $x$ in $G'$ (if such an edge exists).
\end{itemize}

Let $G''$ denote the resulting graph and note that by Claim~A we have $|E(G'')| \geq |E(G')| - |X_2| > 0$. So there exists an edge $uv \in E(G'')$.

\iflong
For every edge, $e$, in $G'$, we say that we swap the edges on the alternating path in $G$ corresponding to $e$ if we add the arcs from $M$ lying on the path to $H$ and remove the arcs from $H$ that lie on the path. We will in the below cases show how to do this such that the resulting digraph is an extension of a \IPNa{} set.
\fi

In all four cases below, which exhaust all possibilities (see Fig.~\ref{Pic3}), we \ifshort can \fi obtain a contradiction (by finding a \IPNa{} set of smaller size than $S$).
\begin{description}
 \item[Case 1\ifshort$(\star)$\fi:] $|\{u,v\} \cap X_2 |=2$. 
 \iflong
 As the edge $uv$ was not deleted from $G'$ when constructing $G''$ we note that 
$d_{H}^-(u) = 2$ and $d_{H}^-(v) = 2$ and there is an edge, $e_u$, in $G'$ between the two vertices in $N_H^-(u)$ and an edge, $e_v$, between the two vertices in $N_H^-(v)$ (See Figure~\ref{Pic3}). 
Swapping edges on the three alternating paths corresponding to the edges $uv$, $e_u$ and $e_v$ in $G'$ (and deleting the arcs into $u$ and $v$ in $H$) 
gives us a an extension of a \IPNa{} set of size smaller than $S$, a contradiction.
\fi

 \item[Case 2\ifshort$(\star)$\fi:] $|\{u,v\} \cap X_2 |=1$. 
 \iflong
 Without loss of generality assume that $u \in X_2$ and $v \not\in X_2$. 
As the edge $uv$ was not deleted from $G'$ when constructing $G''$ we note that                          
$d_{H}^-(u) = 2$ and there is an edge, $e_u$, in $G'$ between the two vertices in $N_H^-(u)$ (See Figure~\ref{Pic3}).
Swapping edges on the two alternating paths corresponding to the edges $uv$ and $e_u$ in $G'$ (and deleting the arcs incident to $u$ and $v$ in $H$) gives us a an extension of a 
\IPNa{} set of size smaller than $S$, a contradiction.
\fi

 \item[Case 3\ifshort$(\star)$\fi:] $|\{u,v\} \cap X_2 |=0$ and $N_H^+(u) = N_H^+(v) \not= \emptyset$. 
 \iflong
 In this case define $x$ such that $N_H^+(u) = \{x\} = N_H^+(v)$. We note
that $d_{H}^-(x) \geq  3$, as the edge $uv$ was not deleted from $G'$ when constructing $G''$ (See Figure~\ref{Pic3}).
 Swapping edges on the alternating path corresponding to the edge $uv$ in $G'$ (and deleting the arcs incident to $u$ and $v$ in $H$) 
gives us a an extension of a \IPNa{} set of size smaller  than $S$, a contradiction.
\fi

 \item[Case 4\ifshort$(\star)$\fi:] $|\{u,v\} \cap X_2 |=0$ and $N_H^+(u) \not= N_H^+(v)$ or $N_H^+(u) = N_H^+(v) = \emptyset$. 
 \iflong
Swapping edges on the alternating path corresponding to the edge $uv$ in $G'$ (and deleting the arcs incident to $u$ and $v$ in $H$) 
gives us a an extension of a \IPNa{} set of size smaller  than $S$, a contradiction.
\fi
\end{description}

\begin{figure}[htb]
\begin{center}
\tikzstyle{vertexWC}=[circle, draw, minimum size=6pt, scale=0.75, inner sep=0.8pt]
\begin{tabular}{|c|c|c|c|} \hline
\begin{tikzpicture}[scale=0.2]
  \node at (1,9) { };

  \node at (4,-1) {(Case 1)};

  \node (x1) at (2,8) [vertexWC]{};
  \node (y1) at (1,1) [vertexWC]{}; \draw[->, line width=0.03cm] (y1) -- (x1);
  \node (y2) at (3,1) [vertexWC]{}; \draw[->, line width=0.03cm] (y2) -- (x1);

  \node (x2) at (6,8) [vertexWC]{};
  \node (y3) at (5,1) [vertexWC]{}; \draw[->, line width=0.03cm] (y3) -- (x2);
  \node (y4) at (7,1) [vertexWC]{}; \draw[->, line width=0.03cm] (y4) -- (x2);

  \draw[dotted, line width=0.03cm] (y1) -- (y2);
  \draw[dotted, line width=0.03cm] (y3) -- (y4);
  \draw[dotted, line width=0.03cm] (x1) -- (x2);

\end{tikzpicture} \hspace{0.2cm} & 
\begin{tikzpicture}[scale=0.20]
  \node at (1,9) { };

  \node at (5,-1) {(Case 2)};

  \node (x1) at (2,8) [vertexWC]{};
  \node (y1) at (1,1) [vertexWC]{}; \draw[->, line width=0.03cm] (y1) -- (x1);
  \node (y2) at (3,1) [vertexWC]{}; \draw[->, line width=0.03cm] (y2) -- (x1);

  \node (x2) at (7,8) [vertexWC]{};
  \node (y3) at (5,1) [vertexWC]{}; \draw[->, line width=0.03cm] (y3) -- (x2);
  \node (y4) at (7,1) [vertexWC]{}; \draw[->, line width=0.03cm] (y4) -- (x2);
  \node (y5) at (9,1) [vertexWC]{}; \draw[->, line width=0.03cm] (y5) -- (x2);

  \draw[dotted, line width=0.03cm] (y1) -- (y2);
  \draw[dotted, line width=0.03cm] (x1) -- (y3);
\end{tikzpicture} \hspace{0.2cm} &
\begin{tikzpicture}[scale=0.20]
  \node at (1,9) { };

  \node at (0.5,-1) {(Case 2)};

  \node (x1) at (2,8) [vertexWC]{};
  \node (y1) at (1,1) [vertexWC]{}; \draw[->, line width=0.03cm] (y1) -- (x1);
  \node (y2) at (3,1) [vertexWC]{}; \draw[->, line width=0.03cm] (y2) -- (x1);

  \node (y3) at (-2,1) [vertexWC]{}; 

  \draw[dotted, line width=0.03cm] (y1) -- (y2);
  \draw[dotted, line width=0.03cm] (x1) -- (y3);
\end{tikzpicture} \hspace{0.2cm} &
\begin{tikzpicture}[scale=0.20]
  \node at (1,9) { };

  \node at (5,-1) {(Case 3)};

  \node (x1) at (4,8) [vertexWC]{};
  \node (y1) at (1,1) [vertexWC]{}; \draw[->, line width=0.03cm] (y1) -- (x1);
  \node (y2) at (3,1) [vertexWC]{}; \draw[->, line width=0.03cm] (y2) -- (x1);
  \node (y3) at (5,1) [vertexWC]{}; \draw[->, line width=0.03cm] (y3) -- (x1);
  \node (y4) at (7,1) [vertexWC]{}; \draw[->, line width=0.03cm] (y4) -- (x1);
  \draw[dotted, line width=0.03cm] (y1) -- (y2);
\end{tikzpicture} \\  \hline
\begin{tikzpicture}[scale=0.20]
  \node at (1,9) { };

  \node at (5.5,-1) {(Case 4)};

  \node (y2) at (2,1) [vertexWC]{}; 

  \node (x2) at (7,8) [vertexWC]{};
  \node (y3) at (5,1) [vertexWC]{}; \draw[->, line width=0.03cm] (y3) -- (x2);
  \node (y4) at (7,1) [vertexWC]{}; \draw[->, line width=0.03cm] (y4) -- (x2);
  \node (y5) at (9,1) [vertexWC]{}; \draw[->, line width=0.03cm] (y5) -- (x2);

  \draw[dotted, line width=0.03cm] (y2) -- (y3);
\end{tikzpicture}  & 
\begin{tikzpicture}[scale=0.20]
  \node at (1,9) { };

  \node at (3,-1) {(Case 4)};

  \node (y2) at (1,1) [vertexWC]{};

  \node (y3) at (5,1) [vertexWC]{}; 
  \draw[dotted, line width=0.03cm] (y2) -- (y3);
\end{tikzpicture}  &  
\multicolumn{2}{c|}{\begin{tikzpicture}[scale=0.20]
  \node at (1,9) { };

  \node at (4.5,-1) {(Case 4)};

  \node (x1) at (1,8) [vertexWC]{};
  \node (y1) at (0,1) [vertexWC]{}; \draw[->, line width=0.03cm] (y1) -- (x1);
  \node (y2) at (2,1) [vertexWC]{}; \draw[->, line width=0.03cm] (y2) -- (x1);

  \node (x2) at (7,8) [vertexWC]{};
  \node (y3) at (5,1) [vertexWC]{}; \draw[->, line width=0.03cm] (y3) -- (x2);
  \node (y4) at (7,1) [vertexWC]{}; \draw[->, line width=0.03cm] (y4) -- (x2);
  \node (y5) at (9,1) [vertexWC]{}; \draw[->, line width=0.03cm] (y5) -- (x2);

  \draw[dotted, line width=0.03cm] (y2) -- (y3);
\end{tikzpicture} } \\  \hline
\end{tabular}
\caption{An illustration of the four cases in the proof of Theorem~\ref{priceOfStabilityI}, where the dotted edges indicate edges in $G'$ and the arcs 
indicate arcs in $H$.}
\label{Pic3}
\end{center}
\end{figure}
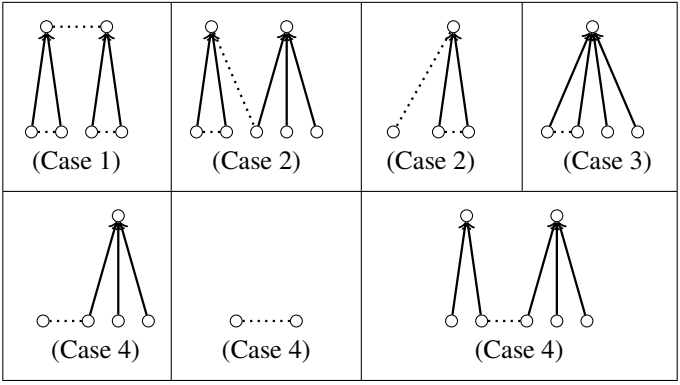

This \iflong contradiction \fi implies that $pn_1(D) > n - \alpha'(D)$ is false, and therefore $pn_1(D) \leq n - \alpha'(D)$, as desired.
\end{proof}

The pure price of anarchy is more straightforward.

\begin{theorem}\label{thm:anarchy}
Fix $k \ge 1$. Then the pure price of anarchy is equal to $k + 1$.
\end{theorem}

\begin{proof}
Consider a complete digraph on $n$ vertices and fix sharing capacity at $k<n$.
The least efficient Nash equilibrium has $n-k$ buyers each choosing the same $k$ non-buyers.
In contrast, at the most efficient outcome each buyer chooses $k$ non-buyers, and so has buyer set of size $\frac{n}{k+1}$, when $n$ is divisible by $k+1$ and round up otherwise.
The ratio of these terms equals the price of anarchy and tends to $k+1$ as $n$ gets large.
Moreover, the price of anarchy is always bounded from above by $k+1$ because there cannot be more than $n$ buyers in any pure Nash equilibrium and there must be at least $n/(k+1)$ buyers.
\end{proof}


\subsection{Efficiency of mixed strategy equilibria}



{Extending our terminology from that used for pure strategy outcomes, the minimum cost of a mixed Nash equilibrium is denoted by $mpn_k(D)$, and the corresponding minimum cost of a mixed buyer assignment \iflong (which does not need to be Nash)\fi is denoted by $mb_k(D)$.}

\iflong
We have the following result on the mixed price of stability.
\fi

\iflong
\begin{theorem} \label{ThmSupportConj}
\fi
\ifshort
\begin{theorem}[$\star$] \label{ThmSupportConj}
\fi
Let $k \geq 1$.
Then the mixed price of stability is equal to $k+1$.
\end{theorem}
\iflong
\begin{proof}
We first show that $mb_k(D) \geq \frac{c \cdot n}{k+1}$. 
Note that any assignment of buyers always gives a penalty of at least  $\frac{c \cdot n}{k+1}$ as 
if anybody does not get chosen then letting them buy will decrease the penalty, so we may assume that 
there are at least $\frac{n}{k+1}$ buyers.
As $mb_k(D)$ is the expected penalty over a given probability space it can not be less than the minimum penalty over any outcome in the probability space.
So, $mb_k(D) \geq \frac{c \cdot n}{k+1}$. 

  In a pure Nash equilibrium every person will have an expected penalty of at most $c$, as if they have an expected penalty of more than this they could 
decrease their penalty by buying (with probability 1). So $mpn_k(D) \leq cn$. This implies the following.

\[
\frac{mpn_k(D)}{mb_k(D)} \leq \frac{cn}{cn/(k+1)} = k+1.
\]

This completes the proof of the first part.  We will now construct a digraph $D_k^r$ for all integers $k \geq 1$ and $r \geq 1$. 
Let $k \geq 1$ and $r \geq 1$ be arbitrary and define $D'$ such that 
 $V(D')=X_1 \cup X_2 \cup \cdots \cup X_{2r+1}$, where each $X_i$ is an independent set of vertices of size $k+1$.
Let $A(D')$ contain all arcs from $X_i$ to $X_{i+1}$ for all $i=1,2,\ldots, 2r+1$ (where indices are taken modulo $2r+1$).
 It is not difficult to see that $D'$ contains a Hamilton cycle $C'$.
Let $D_k^r = D' - A(C')$ and note that $D_k^r$ is $k$-regular.

We may assume that $C'  = c_1 c_2 c_3 \cdots c_{(k+1)(2r+1)} c_1$, where $c_i \in X_i$ for $i=1,2,\ldots,2r+1$.
Let  $S=\{c_1,c_2,\ldots,c_{2r+1},c_{2r+2}\}$ and note that $N_{D_k^r}^+[S]=V(D_k^r)$, so $mb_k(D_k^r) \leq (2r+2)c$. 

Now assume that we have a mixed nash equilibrium, $p$. That is, $p(x)$ denotes the probability that $x$ buys for each $x \in V(D_k^r)$.
First assume that $p(c_1)=1$. This implies that $p(u)=0$ for all $u \in X_2 \setminus \{c_2\}$. This implies that $p(c_3)=1$. Continuing this 
process we see that $p(c_j)=1$ for all $j \in \{1,3,5,\ldots,2r+1, 2r+3 \}$. So, $p(c_{2r+3})=1$, but $c_{2r+3} \in X_2 \setminus \{c_2\}$, so by
the above $p(c_{2r+3})=0$, a contradiction. So $p(c_1)<1$. We can analolously show that $p(x)<1$ for all $x \in V(D_k^r)$.  

We will now show that $p(x)>0$ for all $x \in V(D_k^r)$. Assume this is not the case. As we cannot have $p(x)=0$ for all $x \in V(D_k^r)$, this implies that
there exists an $i$ such that $p(c_i)=0$ and $p(c_{i+1})>0$ (where indices are taken modulo $(k+1)(2r+1)$). Without loss of generality we may assume that $i=1$. 
That is, $p(c_1)=0$ and $p(c_{2})>0$. As $0<p(c_{2})<1$ we note that $c_2$ has probability exactly $1-c$ of being chosen by a buyer. So, some $c_j \in X_1$ has
$p(c_j)>0$. However this implies that $c_{j+1}$ has probability strictly less that $1-c$ of being chosen by a buyer (as $c_j$ will choose $c_2$ with 
positive probability but will not choose $c_{j+1}$ while $c_1$ will not choose either $c_2$ or $c_{j+1}$ as $c_1$ is not buying). But this 
implies that $p(c_{j+1})=1$ in the mixed Nash equilibrium, a contradiction to $p(x)<1$ for all $x \in V(D_k^r)$. So, 
$0<p(x)<1$ for all $x \in V(D_k^r)$ (and if every vertex in $D_k^r$ buys with probability $1-c^{1/k}$ then we obtain a mixed Nash equilibrium).

This implies that every vertex gets a penalty of $c$ and $mpn_k(D_k^r)= c \cdot |V(D_k^r)| = c(2r+1)(k+1)$. Therefore the following holds,

\[
\frac{mpn_k(D_k^r)}{mb_k(D_k^r)} \geq \frac{c(2r+1)(k+1)}{(2r+2)c} = (k+1) \times \left(1 - \frac{1}{2r+2} \right)
\]

Now, as $r \to \infty$ the expression above tends to $k+1$ and the result follows.
\end{proof}
\fi


\section{Discussion}
Our paper provides a complete investigation of the existence, the complexity, and the efficiency of Nash equilibria in public good games with sharing constraints. We focus on directed networks and as we demonstrate this restriction makes a difference for the existence of pure equilibria which always exist in undirected networks~\cite{GerkeGutin:2024:JET,gutin2020uniqueness,GutinNeary:2023:DAM}. In addition, we show that the underlying structure of the network affects the existence of pure Nash equilibria and we provide a dichotomy for this problem. Finally, we obtain bounds, which are tight in many cases, on the price of anarchy and stability; the only open case is stated in Conjecture~\ref{conj1} which we strongly believe is true.





\section*{Acknowledgments}
Argyrios Deligkas acknowledges the support of the EPSRC grant EP/X039862/1.

\bibliography{refsPublicGoods}

\begin{thebibliography}{36}
\providecommand{\natexlab}[1]{#1}

\bibitem[{Allouch(2015)}]{Allouch:2015:JET}
Allouch, N. 2015.
\newblock On the private provision of public goods on networks.
\newblock \emph{Journal of Economic Theory}, 157: 527--552.

\bibitem[{Allouch(2017)}]{Allouch:2017:GEB}
Allouch, N. 2017.
\newblock The cost of segregation in (social) networks.
\newblock \emph{Games and Economic Behavior}, 106: 329 -- 342.

\bibitem[{Baetz(2015)}]{Baetz:2015:TE}
Baetz, O. 2015.
\newblock Social activity and network formation.
\newblock \emph{Theoretical Economics}, 10(2): 315--340.

\bibitem[{Balinski and Ratier(1997)}]{BalinskiRatier:1997:}
Balinski, M.; and Ratier, G. 1997.
\newblock Of Stable Marriages and Graphs, and Strategy and Polytopes.
\newblock \emph{SIAM Review}, 39(4): 575--604.

\bibitem[{Bang-Jensen and Gutin(2008)}]{bang2008digraphs}
Bang-Jensen, J.; and Gutin, G.~Z. 2008.
\newblock \emph{Digraphs: theory, algorithms and applications}.
\newblock Springer Science \& Business Media.

\bibitem[{Bayer, Kozics, and Sz{\H o}ke(2023)}]{BAYER2023105720}
Bayer, P.; Kozics, G.; and Sz{\H o}ke, N.~G. 2023.
\newblock Best-response dynamics in directed network games.
\newblock \emph{Journal of Economic Theory}, 213: 105720.

\bibitem[{Boncinelli and Pin(2012)}]{boncinelli2012stochastic}
Boncinelli, L.; and Pin, P. 2012.
\newblock Stochastic stability in best shot network games.
\newblock \emph{Games and Economic Behavior}, 75(2): 538--554.

\bibitem[{Bramoull{\'e} and Kranton(2007)}]{BramoulleKranton:2007:JET}
Bramoull{\'e}, Y.; and Kranton, R. 2007.
\newblock Public goods in networks.
\newblock \emph{Journal of Economic Theory}, 135(1): 478--494.

\bibitem[{Bramoull{\'e}, Kranton, and
  D'Amours(2014)}]{BramoulleKranton:2014:AER}
Bramoull{\'e}, Y.; Kranton, R.; and D'Amours, M. 2014.
\newblock Strategic Interaction and Networks.
\newblock \emph{American Economic Review}, 104(3): 898--930.

\bibitem[{Chv{\'a}tal(1973)}]{chvatal1973computational}
Chv{\'a}tal, V. 1973.
\newblock On the computational complexity of finding a kernel.
\newblock \emph{Report CRM-300, Centre de Recherches Math{\'e}matiques,
  Universit{\'e} de Montr{\'e}al}, 592.

\bibitem[{Dall'Asta, Pin, and Ramezanpour(2011)}]{dall2011optimal}
Dall'Asta, L.; Pin, P.; and Ramezanpour, A. 2011.
\newblock Optimal equilibria of the best shot game.
\newblock \emph{Journal of Public Economic Theory}, 13(6): 885--901.

\bibitem[{Deligkas et~al.(2023)Deligkas, Fearnley, Hollender, and
  Melissourgos}]{deligkas2023tight}
Deligkas, A.; Fearnley, J.; Hollender, A.; and Melissourgos, T. 2023.
\newblock Tight inapproximability for graphical games.
\newblock In \emph{Proceedings of the AAAI Conference on Artificial
  Intelligence}, volume~37, 5600--5607.

\bibitem[{{Do Dinh} and Hollender(2024)}]{DODINH2024106486}
{Do Dinh}, J.; and Hollender, A. 2024.
\newblock Tight inapproximability of Nash equilibria in public goods games.
\newblock \emph{Information Processing Letters}, 186: 106486.

\bibitem[{Elliott and Golub(2019)}]{ElliottGolub:2019:JPE}
Elliott, M.; and Golub, B. 2019.
\newblock A Network Approach to Public Goods.
\newblock \emph{Journal of Political Economy}, 127(2): 730--776.

\bibitem[{Galeotti et~al.(2010)Galeotti, Goyal, Jackson, Vega-Redondo, and
  Yariv}]{GaleottiGoyal:2010:RES}
Galeotti, A.; Goyal, S.; Jackson, M.~O.; Vega-Redondo, F.; and Yariv, L. 2010.
\newblock Network Games.
\newblock \emph{The Review of Economic Studies}, 77(1): 218--244.

\bibitem[{Gerke et~al.(2024)Gerke, Gutin, Hwang, and
  Neary}]{GerkeGutin:2024:JET}
Gerke, S.; Gutin, G.; Hwang, S.-H.; and Neary, P.~R. 2024.
\newblock Public goods in networks with constraints on sharing.
\newblock \emph{Journal of Economic Theory}, 219: 105844.

\bibitem[{Gilboa(2023)}]{gilboa2023characterization}
Gilboa, M. 2023.
\newblock A Characterization of Complexity in Public Goods Games.
\newblock \emph{arXiv preprint arXiv:2301.11580}.

\bibitem[{Gilboa and Nisan(2022)}]{gilboa2022complexity}
Gilboa, M.; and Nisan, N. 2022.
\newblock Complexity of public goods games on graphs.
\newblock In \emph{International Symposium on Algorithmic Game Theory},
  151--168. Springer.

\bibitem[{Gon{\c c}alves et~al.(2012)Gon{\c c}alves, Havet, Pinlou, and
  Thomass{\'e}}]{DG12}
Gon{\c c}alves, D.; Havet, F.; Pinlou, A.; and Thomass{\'e}, S. 2012.
\newblock On spanning galaxies in digraphs.
\newblock \emph{Discrete Applied Mathematics}, 160(6): 744--754.
\newblock Fourth Workshop on Graph Classes, Optimization, and Width Parameters
  Bergen, Norway, October 2009.

\bibitem[{Gutin, Neary, and Yeo(2020)}]{gutin2020uniqueness}
Gutin, G.; Neary, P.~R.; and Yeo, A. 2020.
\newblock Uniqueness of DP-Nash subgraphs and D-sets in weighted graphs of
  Netflix games.
\newblock In \emph{International Computing and Combinatorics Conference},
  360--371. Springer.

\bibitem[{Gutin, Neary, and Yeo(2023{\natexlab{a}})}]{GutinNeary:2023:DAM}
Gutin, G.; Neary, P.~R.; and Yeo, A. 2023{\natexlab{a}}.
\newblock Exact capacitated domination: on the computational complexity of
  uniqueness.
\newblock \emph{Discrete Applied Mathematics}, 332: 155--169.

\bibitem[{Gutin, Neary, and Yeo(2023{\natexlab{b}})}]{GutinNeary:2023:GEB}
Gutin, G.~Z.; Neary, P.~R.; and Yeo, A. 2023{\natexlab{b}}.
\newblock Unique stable matchings.
\newblock \emph{Games and Economic Behavior}, 141: 529--547.

\bibitem[{Karp(1972)}]{karp1972reducibility}
Karp, R.~M. 1972.
\newblock Reducibility Among Combinatorial Problems.
\newblock In Miller, R.~E.; and Thatcher, J.~W., eds., \emph{Complexity of
  Computer Computations}, 85--103. New York: Plenum Press.

\bibitem[{Kearns, Littman, and Singh(2013)}]{kearns2013graphical}
Kearns, M.; Littman, M.~L.; and Singh, S. 2013.
\newblock Graphical models for game theory.
\newblock \emph{arXiv preprint arXiv:1301.2281}.

\bibitem[{Kinateder and Merlino(2017)}]{KinatederMerlino:2017:AEJM}
Kinateder, M.; and Merlino, L.~P. 2017.
\newblock Public Goods in Endogenous Networks.
\newblock \emph{American Economic Journal: Microeconomics}, 9(3): 187--212.

\bibitem[{Klimm and Stahlberg(2023)}]{klimm2023complexity}
Klimm, M.; and Stahlberg, M.~J. 2023.
\newblock Complexity of equilibria in binary public goods games on undirected
  graphs.
\newblock In \emph{Proceedings of the 24th ACM Conference on Economics and
  Computation}, 938--955.

\bibitem[{Koutsoupias and Papadimitriou(1999)}]{koutsoupias1999worst}
Koutsoupias, E.; and Papadimitriou, C. 1999.
\newblock Worst-case equilibria.
\newblock In \emph{Annual symposium on theoretical aspects of computer
  science}, 404--413. Springer.

\bibitem[{Levit et~al.(2018)Levit, Komarovsky, Grinshpoun, and
  Meisels}]{levit2018incentive}
Levit, V.; Komarovsky, Z.; Grinshpoun, T.; and Meisels, A. 2018.
\newblock Incentive-based search for efficient equilibria of the public goods
  game.
\newblock \emph{Artificial Intelligence}, 262: 142--162.

\bibitem[{L{\'o}pez-Pintado(2013)}]{LOPEZPINTADO2013160}
L{\'o}pez-Pintado, D. 2013.
\newblock Public goods in directed networks.
\newblock \emph{Economics Letters}, 121(2): 160--162.

\bibitem[{Maffray(1992)}]{Maffray:1992:JCTB}
Maffray, F. 1992.
\newblock Kernels in perfect line-graphs.
\newblock \emph{Journal of Combinatorial Theory, Series B}, 1--8.

\bibitem[{Papadimitriou and Peng(2023)}]{PapadimitriouPeng:2023:GEB}
Papadimitriou, C.; and Peng, B. 2023.
\newblock Public goods games in directed networks.
\newblock \emph{Games and Economic Behavior}, 139: 161--179.

\bibitem[{Papadimitriou(1994)}]{PAPADIMITRIOU1994498}
Papadimitriou, C.~H. 1994.
\newblock On the complexity of the parity argument and other inefficient proofs
  of existence.
\newblock \emph{Journal of Computer and System Sciences}, 48(3): 498--532.

\bibitem[{Schulz and Moses(2003)}]{schulz2003performance}
Schulz, A.~S.; and Moses, N. E.~S. 2003.
\newblock On the performance of user equilibria in traffic networks.
\newblock In \emph{SODA}, volume~3, 86--87.

\bibitem[{von Neumann and Morgenstern(1944)}]{NeumannMorgenstern:1944:}
von Neumann, J.; and Morgenstern, O. 1944.
\newblock \emph{Theory of Games and Economic Behavior (60th Anniversary
  Commemorative Edition)}.
\newblock Princeton University Press.

\bibitem[{Yang and Wang(2020)}]{yang2020refined}
Yang, Y.; and Wang, J. 2020.
\newblock A refined study of the complexity of binary networked public goods
  games.
\newblock \emph{arXiv preprint arXiv:2012.02916}.

\bibitem[{Yu et~al.(2020)Yu, Zhou, Brantingham, and
  Vorobeychik}]{yu2020computing}
Yu, S.; Zhou, K.; Brantingham, J.; and Vorobeychik, Y. 2020.
\newblock Computing equilibria in binary networked public goods games.
\newblock In \emph{Proceedings of the AAAI Conference on Artificial
  Intelligence}, volume~34, 2310--2317.

\end{thebibliography}


\end{document}